\documentclass[10pt,a4paper,british]{IEEEtran}
\usepackage{mathptmx}
\usepackage[T1]{fontenc}
\usepackage[latin9]{inputenc}
\usepackage{amsmath}
\usepackage{graphicx}
\usepackage{amssymb}

\makeatletter

\providecommand{\tabularnewline}{\\}

\newtheorem{definitn}{Definition}
\newtheorem{remrk}{Remark}
\newtheorem{lemma}{Lemma}
\newtheorem{example}{Example}
\newtheorem{prop}{Proposition}
\newtheorem{thm}{Theorem}

\usepackage{pslatex}
\usepackage{url}
\usepackage{subfigure}
\numberwithin{thm}{section}
\numberwithin{lemma}{section}
\numberwithin{prop}{section}

\usepackage{babel}
\makeatother

\begin{document}

\title{Asymptotic stability and capacity results for a broad family of power
adjustment rules: Expanded discussion }

\author{Virgilio Rodriguez$^{1}$, Rudolf Mathar$^{1}$ and Anke Schmeink$^{2}$
\\
 $^{1}$Theoretische Informationstechnik\\
$^{2}$UMIC Research Centre\\
 RWTH Aachen \\
 Aachen, Germany\\
email: vr@ieee.org, \{mathar@ti , schmeink@umic\}.rwth-aachen.de}

\maketitle
\begin{abstract}
In any wireless communication environment in which a transmitter creates
interference to the others, a system of non-linear equations arises.
Its form (for 2 terminals) is p1=g1(p2;a1) and p2=g2(p1;a2), with
p1, p2 power levels; a1, a2 quality-of-service (QoS) targets; and
g1, g2 functions akin to \char`\"{}interference functions\char`\"{}
in Yates (JSAC, 13(7):1341-1348, 1995). Two fundamental questions
are: (1) does the system have a solution?; and if so, (2) what is
it?. (Yates, 1995) shows that IF the system has a solution, AND the
{}``interference functions'' satisfy some simple properties, a {}``greedy''
power adjustment process will always converge to a solution. We show
that, if the power-adjustment functions have similar properties to
those of (Yates, 1995), and satisfy a condition of the simple form
gi(1,1,...,1)<1, then the system has a unique solution that can be
found iteratively. As examples, feasibility conditions for macro-diversity
and multiple-connection receptions are given. Informally speaking,
we complement (Yates, 1995) by adding the feasibility condition it
lacked. Our analysis is based on norm concepts, and the Banach's contraction-mapping
principle.
\end{abstract}
\IEEEpeerreviewmaketitle

\section{Introduction}

In any wireless communication environment in which a terminal creates
interference to the others, a system of non-linear equations (or more
generally inequalities) arises. It can be written as $p_{i}=g_{i}(\mathbf{p}_{-i};\alpha_{i})$
for $i=1,\cdots,N$, where $g_{i}$ is an appropriate function, $\alpha_{i}$
is a quality-of-service (QoS) target, and $\mathbf{p}_{-i}$ denotes
the vector of the power levels of the other terminals. Two fundamental
questions immediately arise: (i) does the system have solutions? (i.e.,
are the QoS targets {}``feasible''?); and if so, (ii) what is one
such solution? 

The feasibility question is critical, because if a set of terminals
is admitted into service when their QoS targets are {}``infeasible'',
valuable resources (e.g. time and energy) may be wasted in a futile
search for a power vector that does not actually exist. Thus, a formula
that can directly determine whether a set of QoS targets are feasible
has an evident practical utility: admission control. For example,
for the specific case of a CDMA wireless communication system in which
base stations {}``cooperate'', \cite{powcon_macrod_cap_hanly96}
shows that --- with some restrictions --- the QoS targets are feasible
if their sum is less than the number of receivers. The set of all
the QoS vectors that can be accommodated are associated with the {}``capacity
region'' of the system. 

An exact closed-form answer to the second question is available only
for very simple scenarios, such as the reverse link of an isolated
CDMA cell. However, the pertinent power vector may be found iteratively,
in which case, 2 other key questions arise: (i) does the chosen power
adjustment algorithm converge?, and if so, (ii) to the same point,
regardless of the initial powers? (i.e., is the process asymptotically
stable?). 

Reference \cite{powcon_stab_yates95} studies the convergence of a
{}``greedy'' power adjustment process --- terminals take turns,
each choosing a power level in order to achieve its desired QoS while
taking the other power levels as fixed --- within an abstract model
that {}``hides'' all details of the physical system inside the power-adjustment
functions, which are assumed to have certain simple properties. This
approach is important because its results apply to all practical systems
that can be shown to satisfy the assumed properties. Reference \cite{powcon_stab_yates95}
shows that if the {}``interference function'' is non-negative, non-decreasing,
and --- in certain sense --- (sub)homogeneous, greedy power adjustment
converges to a unique vector, \emph{provided} that the underlying
QoS targets are\emph{ feasible}. Recently, several publications have
revisited \cite{powcon_stab_yates95} with various aims. Reference
\cite{powcon_stab_discr_sung99} introduces and establishes the convergence
of an algorithm that can handle the discreteness (quantisation) of
power adjustment typical of practical systems, a case that does not
fit into the framework of \cite{powcon_stab_yates95}. Reference \cite{powcon_stab_discr_leung04}
extends \cite{powcon_stab_discr_sung99} to establish the convergence
of a {}``canonical class'' of algorithms, which includes many algorithms
previously proposed in the scientific literature. Opportunistic communication
as appropriate for delay-tolerant traffic is the focus of \cite{powcon_stab_oppor_sung05}.
Reference \cite{powcon_stab_schubert05} models interference within
an axiomatic framework and characterises the feasible quality-of-service
region corresponding to the max-min signal-to-interference balancing
problem. 

However, neither \cite{powcon_stab_yates95} nor its descendants provide
a general feasibility condition for their respective family of functions.
The present work adds sub-additivity to the properties of \cite{powcon_stab_yates95},
and this, in turn, leads to the simple feasibility condition $f_{i}(1,1,\ldots,1)<1$,
which also works without sub-additivity, but is then more conservative.
Particularised to \cite{powcon_macrod_cap_hanly96}, our result leads
to a still simple but more sophisticated macro-diversity capacity
formula that --- through a dependence on \emph{relative} channel gains
--- sensibly adjusts itself to the realistic situation in which each
terminal is in range of only a subset of the receivers, as discussed
further below and in \cite{powcon_macrod_cap_rodriguez08}. 

To obtain our result, we explore the convergence of a class of adjustment
functions of the form $f_{i}(\mathbf{p}_{-i})+c_{i}$ where $c_{i}\geq0$
and $f_{i}$ belongs to a large family of which {}``norms'' are
a special case. A norm is an intuitive generalisation of the length
of a vector. If $f_{i}$ is a norm, $f_{i}(\mathbf{p}_{-i})+c_{i}$
can be interpreted as a sum of {}``noise'' plus the {}``size''
of the interfering power, with the terminal adjusting its power to
keep the carrier-to-noise-plus-interference ratio near a target value.
The mathematical properties of norms and, more generally, semi-norms
are well-understood, and have proven useful in many contexts (for
an interesting application to beam forming see \cite{powcon_beam_norm_schubert04}).
As an added bonus, these functions are convex (and hence continuous),
which is often a desirable property.

Our related work \cite{powcon_stab_norm_rodriguez09,powcon_stab_homog_rodriguez09}
introduces a more concrete model that explicitly considers details
such as channel gains, QoS constraints, and the number of receivers.
The more abstract {}``high-level view'' of \cite{powcon_stab_yates95}
and its descendants ---including the present paper --- is evidently
the most general; however the {}``lower-level view'' of the concrete
model may provide insights and opportunities otherwise unavailable
(e.g., \cite{powcon_stab_homog_rodriguez09} provides in closed-form
a general \emph{conservative} solution to the system of equations
under study). Thus, both approaches are complementary. 

In the next section, we state and discuss our main results. Then,
we formally specify the properties of the functions we study. Subsequently,
we utilise the contraction mapping technique to characterise the conditions
leading to the convergence of an adjustment process done with these
functions. Then, we connect these conditions to the quality-of-service
requirements of the terminals. Subsequently, we apply our analysis
to other families of functions, including those studied by \cite{powcon_stab_yates95,powcon_stab_schubert05}.
Two appendices provide essential mathematical background, and specific
key results from the literature.

\section{Main result, applications and extensions }

\label{sec:powcon-stab-adj-disc-main}

In this section we \emph{informally} state our main result, discuss
how it can be directly applied, the methodology leading to it, and
how to extend it to cover functions satisfying other sets of axioms%
\footnote{Some readers may prefer to leave this section for last.%
}. We also discuss the macro-diversity capacity result it yields, and
compare this result to that provided by \cite{powcon_macrod_cap_hanly96}.

\subsection{Main result}

\label{sub:powcon-stab-adj-disc-reslt}

A function $f$ is \emph{quasi-semi-normal} if it has four basic properties
formally stated in Definition \ref{def:powcon-stab-adj-qseminormal}:
non-negativity, monotonicity, sub-homogeneity of degree one and sub-additivity
(the triangle inequality). With $\vec{1}$ denoting the {}``all-ones''
vector of appropriate length, our main result, Theorem \ref{thm:powcon-stab-adj-contr-cond},
can be \emph{informally} re-stated as: 

\begin{quotation}
\noindent If $f_{i}$ is quasi-semi-normal and satisfies $f_{i}(\vec{1})<1$
then the adjustment process defined by $p_{i}(t+1)=f_{i}(\mathbf{p}_{-i}(t))+c_{i}$
($c_{i}\geq0$) converges to the same vector $\mathbf{p}^{*}$, regardless
of the initial power levels.
\end{quotation}

\subsection{Direct applicability}

\label{sub:powcon-stab-adj-disc-appl}

$f_{i}$ generally depends on the terminals' quality-of-service (QoS)
parameters. Thus, from the set of conditions $f_{i}(\vec{1})<1$ one
can determine whether a given QoS vector is {}``feasible'' in the
sense that it leads to a convergent power-adjustment process. This
information answers the important admission-control question: can
the system admit a terminal that wishes service at a given QoS level,
and satisfy the QoS requirements of the new and the incumbent terminals? 

Theorem \ref{thm:powcon-stab-adj-contr-cond} can be very useful,
because a very large family of functions satisfies Definition \ref{def:powcon-stab-adj-qseminormal}.
This includes the sub-family of parametric H\"{o}lder norms (which
itself includes the Euclidean norm, the {}``max'' norm, as well
as the sum-of-absolute-values norm as special cases) and \emph{all}
other (semi-)norms. Furthermore, it is possible to define new (semi-)norms,
by performing simple operations on known ones; e.g., the sum or maximum
of two norms is a new norm, and if $f(\cdot)$ is a norm and $M$
is a suitably dimensioned non-singular matrix, then $f(M\cdot)$ defines
also a norm \cite[Sec. 5.3]{math_alg_lin_matrix_mg_horn85}. 

One can envision three general use-cases for Theorem \ref{thm:powcon-stab-adj-contr-cond}:
(i) the system's most {}``natural'' power adjustment process fits
the pattern $p_{i}=f_{i}(\mathbf{p}_{-i})+c_{i}$ with $f_{i}$ a
quasi-semi-normal function and $c_{i}\geq0$ (e.g., the fixed assignment
scenario of \cite{powcon_stab_yates95}) (ii) the engineer can freely
choose the terminals' power-adjustment rules (in which case the family
of functions under study is sufficiently large to give the engineer
wide latitude in making this choice), (iii) the engineer can analyse
the system under an adjustment rule that has the desired properties,
and  overestimates the {}``true'' terminal's power needs, which
leads to a conservative admission policy (as will be discussed further
and illustrated below).

\subsection{Methodology}

\label{sub:powcon-stab-adj-disc-method}

We obtain our results through fixed-point theory. One can formally
describe the power adjustment process through a \emph{transformation}
$\mathbf{T}$ that takes as input a power vector $\mathbf{p}$ and
{}``converts'' it into a new one, $\mathbf{T}(\mathbf{p})$. The
limit of the adjustment process, if any, is a vector that satisfies
$\mathbf{p^{*}}=\mathbf{T}(\mathbf{p}^{*})$. For a transformation
$\mathbf{T}$ from certain space into itself, fixed-point theory provides
conditions under which $\mathbf{T}$ has a {}``fixed-point'', that
is, there is a point $\mathbf{x}^{*}$ in the concerned space such
that $\mathbf{x}^{*}=\mathbf{T}(\mathbf{x}^{*})$. In particular,
Theorem \ref{thm:math-fixpt-banach1922} holds that, if $\mathbf{T}$
is a \emph{contraction} (Definition \ref{def:math-contract-map}),
then $\mathbf{T}$ has a \emph{unique} fixed-point, and that it can
be found iteratively via successive approximation (Definition \ref{def:math-succes-approx}),
irrespective of the starting point. Theorem \ref{thm:powcon-stab-adj-contr-cond}
identifies conditions under which the transformation of interest is
a contraction. The core of its proof has three simple steps, and each
directly follows from exactly one of the properties of the functions
we study.

\subsection{Applicability to other function families}

\label{sub:powcon-stab-adj-disc-other}

If one knows that an adjustment rule fails to satisfy Definition \ref{def:powcon-stab-adj-qseminormal},
but otherwise has certain {}``nice'' properties, two relevant and
fair question are: (a) does it always exist a corresponding adjustment
rule that overestimates a terminal's power needs, and that has the
necessary properties for the applicability of Theorem \ref{thm:powcon-stab-adj-contr-cond}?,
and (b) can such rule be identified in general, in terms of the original
function? The answers will, of course, depend on which are the properties
that the original adjustment rule does posses. 

\begin{table*}
\caption{\label{tab:powcon-stab-fws} Selected power-adjustment frameworks
compared}

\noindent \begin{centering}
\begin{tabular}{|c|c|c|c|c|}
\hline 
Framework & Monotonicity & Homogeneity & Sub-additivity & Feasibility\tabularnewline
\hline
\hline 
Yates\cite{powcon_stab_yates95} & $\mathbf{x}\geq\mathbf{y}\implies\mathbf{f}(\mathbf{x})\geq\mathbf{f}(\mathbf{y})$ & $\lambda>1\implies\mathbf{f}(\lambda\mathbf{x})<\lambda\mathbf{f}(\mathbf{x})$ & --- & ---\tabularnewline
\hline 
S-B\cite{powcon_stab_schubert05} & $\mathbf{x}\geq\mathbf{y}\implies f(\mathbf{x})\geq f(\mathbf{y})$ & $\lambda\geq0\implies f(\lambda\mathbf{x})=\lambda f(\mathbf{x})$ & --- & ---\tabularnewline
\hline 
Herein & $f(\mathbf{x})\leq f(\left\Vert \mathbf{x}\right\Vert _{\infty}\vec{1})$ & $\lambda\in(0,1)\implies f(\lambda\vec{1})\leq\lambda f(\vec{1})$ & $f(\mathbf{x}+\mathbf{y})\leq f(\mathbf{x})+f(\mathbf{y})$ & $f(\vec{1})\leq1$\tabularnewline
\hline
\end{tabular}
\par\end{centering}
\end{table*}

While ignoring certain technical subtleties, Table \ref{tab:powcon-stab-fws}
compares the properties assumed herein to those assumed by \cite{powcon_stab_yates95}
and \cite{powcon_stab_schubert05}. Non-negativity is an imposition
of the physical world that applies to all axiomatic frameworks. Additionally,
the three frameworks assume a form of monotonicity and homogeneity
({}``scalability''). Unique to the present contribution is the triangle
inequality, which in turn leads to a simple feasibility result not
available under other frameworks. This comparison suggests that, besides
non-negativity, monotonicity and some form of homogeneity be the {}``nice''
properties to be kept.

\subsubsection{Homogeneity notions}

\label{sub:powcon-stab-adj-disc-other-homog}

The homogeneity axioms displayed in Table \ref{tab:powcon-stab-fws}
exhibit noticeable differences. Whereas \cite{powcon_stab_schubert05}'s
homogeneity applies to all scaling constants, $\lambda$, our axiom
applies to $\lambda$ in $(0,1)$. However, by Lemma \ref{lem:powcon-stab-adj-qsemi-triang-subhom},
homogeneity for $\lambda\in(0,1)$ together with sub-additivity imply
homogeneity for all positive $\lambda$.

In \cite{powcon_stab_yates95} the considered functions are strictly
sub-homogeneous, but only for $\lambda>1$. However, \cite{powcon_stab_yates95}'s
{}``interference functions'' include additive {}``noise''. By
contrast, our functions have the form $f_{i}(\mathbf{x}_{-i})+c_{i}$,
and homogeneity applies to $f_{i}$ \emph{only}. If $f(x)=g(x)+c$
where $g(\lambda x)=\lambda g(x)$ and $c>0$, then $f$ is strictly
sub-homogeneous of degree one, but only for $\lambda>1$ {[}$f(\lambda x)\equiv g(\lambda x)+c=\lambda g(x)+c\equiv\lambda f(x)+(1-\lambda)c$;
thus, $f(\lambda x)<\lambda f(x)$ for $\lambda>1$]. 

Thus, while the homogeneity assumptions of \cite{powcon_stab_yates95},
\cite{powcon_stab_schubert05} and ours are \emph{not} technically
\emph{equivalent}, they are, to some extend, mutually consistent.
On the other hand, our functions only need homogeneity at the point
$\mathbf{x}=\vec{1}$.

\subsubsection{(sub)Homogeneous adjustment processes}

\label{sub:powcon-stab-adj-disc-other-boche}

Subsection \ref{sub:powcon-stab-adj-feas-xtens-boche} shows that
if the original adjustment function $f$ fails to satisfy the triangle
inequality, but that it is however monotonically non-decreasing and
(sub)homogeneous of degree one for any positive constant (which is
satisfied by all the functions considered by \cite{powcon_stab_schubert05},
for example), then $\phi(\mathbf{x}):=\left\Vert \mathbf{x}\right\Vert _{\infty}f(\vec{1})$
{}``dominates'' $f$ ($f(\mathbf{x})\leq\phi(\mathbf{x})$ everywhere),
and has the desired properties (because $\phi$ is a scaled version
of the norm $\left\Vert \mathbf{x}\right\Vert _{\infty}=\max(x_{1},\cdots,x_{N})$).
Thus, one can obtain a conservative admission rule by applying Theorem
\ref{thm:powcon-stab-adj-contr-cond} to an adjustment process in
which terminal $i$ updates its power with $\phi_{i}(\mathbf{x}):=\left\Vert \mathbf{x}\right\Vert _{\infty}f_{i}(\vec{1})$.
The appropriate feasibility condition is $\phi_{i}(\vec{1})\equiv\bigl\Vert\vec{1}\bigr\Vert_{\infty}f_{i}(\vec{1})=f_{i}(\vec{1})<1$.
Thus, $f_{i}(\vec{1})<1$ also works for the original process. However,
in this case the condition is more conservative than it would be,
if the original $f_{i}$ also satisfies sub-additivity, because now
the condition has been obtained through the dominating $\phi_{i}$. 

By exploiting the known special structure of the original adjustment
function, one may be able to obtain a {}``tighter bound'' than $\phi_{i}$.
In fact, that is how we have approached macro-diversity. Nevertheless,
it is useful and comforting to know that for a very large family of
functions, the construction $\phi_{i}$ leads to one simple capacity
result, when no better such result is available.

\subsubsection{Partially sub-homogeneous adjustment processes}

\label{sub:powcon-stab-adj-disc-other-yates}

Not \emph{every} function that satisfies \cite{powcon_stab_yates95}'s
axioms can be written as the sum of a positive constant and a function
that is homogeneous of degree one (see subsection \ref{sub:powcon-stab-adj-disc-other-homog}).
Nevertheless, subsection \ref{sub:powcon-stab-adj-feas-xtens-yates}
shows that the adjustment process corresponding to each of the models
cited by \cite{powcon_stab_yates95} (i) has the form assumed in the
present work, or (ii) can be handled through a special bounding function,
or (iii) --- under the mild assumption that random noise is negligible
--- is covered by the discussion in sub-section \ref{sub:powcon-stab-adj-disc-other-boche}.
One of \cite{powcon_stab_yates95}'s examples is macro-diversity ---
discussed at length throughout the present work ---, while the multiple
connection (MC) scenario is discussed in some detail in section \ref{sub:powcon-stab-adj-feas-xtens-yates}.

\subsection{The case of macro-diversity}

\label{sub:powcon-stab-adj-disc-macrod}

With macro-diversity, the cellular structure of a wireless communication
network is removed and each terminal is jointly decoded by all receivers
in the network \cite{powcon_phd_hanly93,powcon_macrod_cap_hanly96}.
Macro-diversity is interesting because it can increase the capacity
of a wireless cellular network, and mitigate shadow fading. As a proof-of-concept
scenario, we have applied Theorem \ref{thm:powcon-stab-adj-contr-cond}
to macro-diversity, and obtained a \emph{new} simple closed-form feasibility
condition, (\ref{eq:powcon-stab-adj-exa-macrod-feas-best}), which
has a number of advantages over that previously available. For a macro-diversity
system with $K$ receivers, and $N$ terminals operating on the reverse
link, where $\boldsymbol{\alpha}:=(\alpha_{1},\cdots,\alpha_{N})$
is the vector of desired carrier-to-interference ratios, $h_{i,k}$
the channel gain in the signal from terminal $i$ arriving at receiver
$k$, and $g_{i,k}=h_{i,k}/h_{i}$ with $h_{i}=h_{i,1}+\cdots+h_{i,K}$,
Theorem \ref{thm:powcon-stab-adj-contr-cond} dictates that: 

\begin{quotation}
\noindent if at each receiver $k$ and for each terminal $i$, $\sum_{n\neq i}\alpha_{n}g_{n,k}<1$
then it is possible for each terminal $i$ to operate at the CIR $\alpha_{i}$. 
\end{quotation}
Thus, the greatest weighted sum of $N-1$ carrier-to-interference
ratios must be less than 1, in order for $\boldsymbol{\alpha}$ to
lie in the {}``capacity region'' of the system. The weights are
relative channel gains. At most $NK$ such simple sums need to be
checked before an admission decision. 

Condition (\ref{eq:powcon-stab-adj-exa-macrod-feas-best}) is closest
to that provided by \cite{powcon_macrod_cap_hanly96} in the special
case in which each terminal is {}``equidistant'' from each receiver;
that is, for each $i$, $h_{i,k}\approx h_{i,l}\quad\forall k,l$
(for example, the terminals may be distributed along a line that is
perpendicular to the axis between the 2 symmetrically placed receivers).
In this case, each $g_{n,k}=1/K$, and condition (\ref{eq:powcon-stab-adj-exa-macrod-feas-best})
reduces to $\sum_{\substack{n=1\\
n\neq i}
}^{N}\alpha_{n}<K$ for each $i$ (which is consistent with condition (\ref{eq:powcon-stab-adj-exa-simp-feas-best}),
for $K=1$). $\sum_{\substack{n=1\\
n\neq i}
}^{N}\alpha_{n}$ adds all $\alpha_{i}$ except one; such sum is, evidently, largest
when it leaves out the smallest $\alpha_{i}$. By comparison, \cite{powcon_macrod_cap_hanly96}
gives the condition $\sum_{n=1}^{N}\alpha_{n}<K$ for all cases. Condition
(\ref{eq:powcon-stab-adj-exa-macrod-feas-best}) is the least conservative
of the two because it leaves out one $\alpha_{i}$ (the smallest)
from the sum. For 3 terminals and 2 receivers, the original yields
the symmetric pyramidal region with vertexes (0,0,0), (2,0,0), (0,2,0)
and (0,0,2) shown in darker colour in fig.\ \ref{fig:macrocap_sym_new_vs_orig}.
By contrast, $\sum_{\substack{n=1\\
n\neq i}
}^{3}\alpha_{n}<2$ --- to which condition (\ref{eq:powcon-stab-adj-exa-macrod-feas-best})
reduces, in this example --- yields a capacity region that completely
contains the darker triangular pyramid, and extends to include the
grayish triangular volume limited above by the line segment between
(0,0,2) and (1,1,1) (indeed, the point (0,99 , 0,99 , 0,99) does satisfy
$\sum_{\substack{n=1\\
n\neq i}
}^{3}\alpha_{n}<2$ but definitely \emph{not} $\alpha_{1}+\alpha_{2}+\alpha_{3}<2$ ). 

\begin{figure}
\begin{centering}
\includegraphics[scale=0.4]{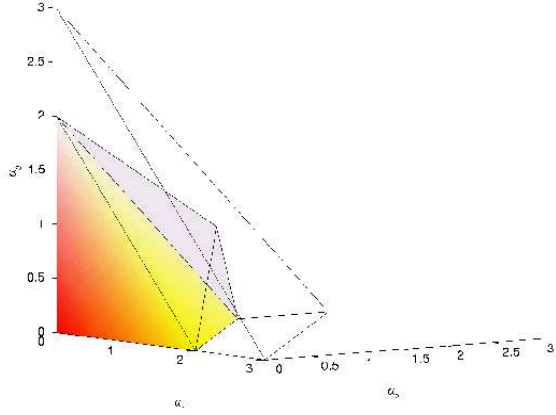}
\par\end{centering}

\caption{\label{fig:macrocap_sym_new_vs_orig} 3 terminals {}``equidistant''
from each of 2 receivers: the original limits capacity to the darker
pyramid, while {}``true'' capacity also includes the grayish triangular
volume. If the terminals cannot be {}``heard'' by a 3rd receiver,
the original greatly overestimates the capacity region by expanding
it to the outer pyramid.}

\end{figure}

It is also significant that the channel gains completely drop out
of the condition given by \cite{powcon_macrod_cap_hanly96}. This
fact reduces somewhat the complexity of the condition. Yet some reflection
suggests that an admission decision should be influenced by the location
of the incumbent and entering terminals. For example, if most active
terminals are near a few receivers, then it should make a difference
to the system whether a new terminal wants to join the crowded region,
or a distant less congested area. Because the original condition is
independent of the channel gains, and hence of the terminals' locations,
it cannot adapt to special geographical distributions of the terminals.
Thus, the original may yield over-optimistic results under certain
channel states, such as when most terminals are in effective range
of only a few receivers. For instance, suppose in the previous example
that a third receiver exists, but that the terminals are located in
such a way that, for each $i$, $h_{i,1}\approx h_{i,2}$ while $h_{i,3}\approx0$.
Thus, $g_{i,1}\approx g_{i,2}\approx1/2$ and $g_{i,3}\approx0$.
Then, condition (\ref{eq:powcon-stab-adj-exa-macrod-feas-best}) still
reduces to $\sum_{\substack{n=1\\
n\neq i}
}^{N}\alpha_{n}<2$ for each $i$, and leads to the already discussed capacity region.
However, the original condition yields $\sum_{n=1}^{N}\alpha_{n}<3$,
which, as illustrated by fig.\ \ref{fig:macrocap_sym_new_vs_orig},
greatly overestimate the capacity region, by extending it to the outer
triangular pyramid with vertexes (0,0,0), (3,0,0), (0,3,0) and (0,0,3)).

Let us now consider the simple \emph{asymmetric} case of 3 terminals
and 2 receivers, with relative gains to the first receiver of 2/3,
1/3, and 1/2, respectively. Condition (\ref{eq:powcon-stab-adj-exa-macrod-feas-best})
leads to 3 inequalities per receiver, such as $\frac{2}{3}\alpha_{1}+\frac{1}{3}\alpha_{2}<1$,
$\frac{1}{3}\alpha_{1}+\frac{2}{3}\alpha_{2}<1$, $\frac{2}{3}\alpha_{2}+\frac{1}{2}\alpha_{3}<1$,
etc. The combination of these inequalities yields a region illustrated
by fig.\ \ref{fig:macrocap_exct_asym_3x2_per}, which is limited
from above by the line segment between (0,0,2) and (1,1,2/3). As already
discussed, the result from \cite{powcon_macrod_cap_hanly96}, $\sum_{n=1}^{N}\alpha_{n}<2$,
yields a symmetric pyramidal region with vertexes at (0,0,0), (2,0,0),
(0,2,0) and (0,0,2) which, as illustrated by fig.\ \ref{fig:macrocap_exct_han_asym_3x2},
intersects with --- but neither contains nor is contained by --- the
region described by fig.\ \ref{fig:macrocap_exct_asym_3x2_per}.

\begin{figure}
\begin{centering}
\subfigure["True" capacity region]{\label{fig:macrocap_exct_asym_3x2_per}\includegraphics[scale=0.4]{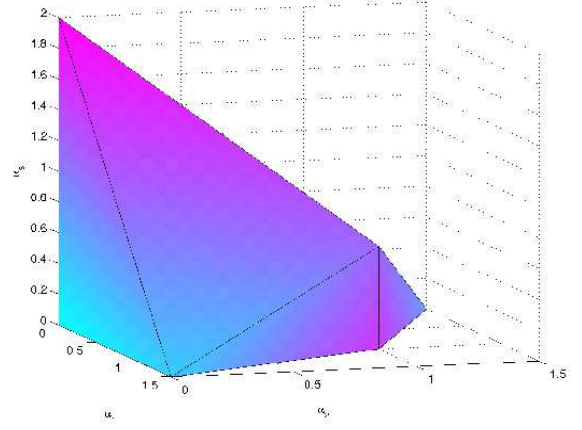}}
\par\end{centering}

\noindent \begin{centering}
\subfigure[The original formula produces a capacity region that neither includes nor is included by the "true" region]{\label{fig:macrocap_exct_han_asym_3x2}\includegraphics[scale=0.4]{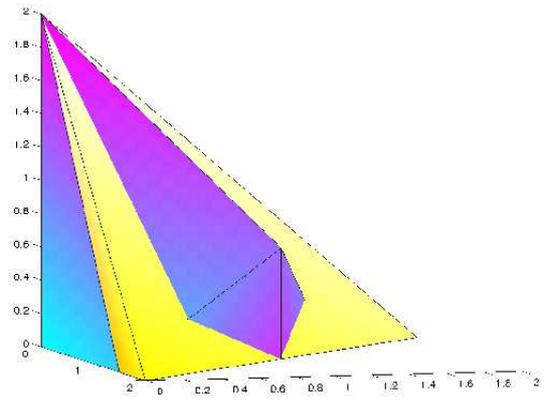}}
\par\end{centering}

\caption{\label{fig:macrocap_exct_asym} The macro-diversity capacity region
for 3 terminals and 2 receivers, for specific asymmetric channel gains}

\end{figure}

As discussed further in \cite{powcon_macrod_cap_rodriguez08}, condition
(\ref{eq:powcon-stab-adj-exa-macrod-feas-best}) yields a low-complexity
algorithm for admission-control decisions, which adapts itself in
a sensible manner to special channel states. Channel gains also play
a prominent role in the feasibility analysis of other multi-cell CDMA
systems, such as in \cite{powcon_cdma_capreg_catrein04}.


\section{A class of sub-additive adjustment rules}

\label{sec:powcon-stab-adj-props}

We focus below on the properties of the individual adjustment function.
Thus, from the standpoint of \cite{powcon_stab_yates95}, our focus
is $I_{i}(\mathbf{p})$, a component of $\mathbf{I}(\mathbf{p})$.

\subsection{Definition and basic properties}

\label{sec:powcon-stab-adj-props-bas}

Below, $\Re_{+}^{M}$ denotes the non-negative orthan of $M$-dimensional
Euclidean space. $\vec{1}_{M}$ denotes the element of $\Re^{M}$
with each component equal to one (the sub-index may be omitted when
appropriate). $\mathbb{N}=\{1,2,\ldots\}$ (the set of Natural numbers). 

We study adjustment rules of the general form $f_{i}(\mathbf{p}_{-i})+c_{i}$
where $c_{i}\in\Re_{+}$and $f_{i}$ is \emph{quasi-semi-normal}.

\begin{definitn}
\label{def:powcon-stab-adj-qseminormal} A function $f:\Re^{M}\rightarrow\Re$
is \emph{quasi-semi-normal} if it satisfies\begin{eqnarray}
f(\mathbf{x})\geq0 &  & \forall\mathbf{x}\in\Re^{M}\label{eq:powcon-stab-adj-qsemin-pos}\\
f(\lambda\vec{1})\leq\lambda f(\vec{1}) &  & \forall\lambda\in(0,1)\label{eq:powcon-stab-adj-qsemin-subhom}\\
f(\mathbf{x}+\mathbf{y})\leq f(\mathbf{x})+f(\mathbf{y}) &  & \forall\mathbf{x},\mathbf{y}\in\Re^{M}\label{eq:powcon-stab-adj-qsemin-triang}\\
f(\mathbf{x})\leq f(\left\Vert \mathbf{x}\right\Vert _{\infty}\vec{1}_{M}) &  & \forall\mathbf{x}\in\Re^{M}\label{eq:powcon-stab-adj-qsemin-mon}\end{eqnarray}
 
\end{definitn}
The preceding conditions are often associated with the words or phrases:
non-negativity (\ref{eq:powcon-stab-adj-qsemin-pos}), sub-homogeneity
(\ref{eq:powcon-stab-adj-qsemin-subhom}), sub-additivity or {}``the
triangle inequality'' (\ref{eq:powcon-stab-adj-qsemin-triang}),
and monotonicity (\ref{eq:powcon-stab-adj-qsemin-mon}).

\begin{remrk}
\label{rem:powcon-stab-adj-qsemi-convex} Below, we only need our
functions to satisfy $f(\lambda\mathbf{x})\leq\lambda f(\mathbf{x})$
at $\mathbf{x}=\vec{1}$. If a function satisfies over its entire
domain both (\ref{eq:powcon-stab-adj-qsemin-triang}) and $f(\lambda\mathbf{x})\leq\lambda f(\mathbf{x})$,
then it is convex (see also Remark \ref{rem:math-norm-convex}). 
\end{remrk}

\begin{remrk}
\label{rem:powcon-stab-adj-qsemi-nonpos-vects} Although power vectors
are inherently non-negative, the difference between 2 non-negative
vectors can, evidently, have negative components. Thus, certain properties
in Definition \ref{def:powcon-stab-adj-qseminormal} must consider
vectors that have negative components.
\end{remrk}

\begin{remrk}
\label{rem:powcon-stab-adj-trian-contin} By Lemma \ref{lem:math-triang-revrs},
a function that satisfies condition (\ref{eq:powcon-stab-adj-qsemin-triang})
also satisfies $\vert f(\mathbf{x})-f(\mathbf{y})\vert\leq f(\mathbf{x}-\mathbf{y})$
, the {}``reverse'' triangle inequality. 
\end{remrk}

\begin{remrk}
\label{rem:powcon-stab-adj-qsemi-triang-subhom-nat} With $\mathbf{x}=\mathbf{y}$
in condition (\ref{eq:powcon-stab-adj-qsemin-triang}) one concludes
that $f(2\mathbf{x})\leq2f(\mathbf{x})$, which easily extends to
$f(m\mathbf{x})\leq mf(\mathbf{x})$ for any $m\in\mathbb{N}$.
\end{remrk}

\begin{remrk}
\label{rem:powcon-stab-adj-qsemi-mono} In (\ref{eq:powcon-stab-adj-qsemin-mon}),
the vector $\left\Vert \mathbf{x}\right\Vert _{\infty}\vec{1}_{M}$
is obtained from $\mathbf{x}$ by replacing each of its components
with the largest of the absolute values of these components, $\left\Vert \mathbf{x}\right\Vert _{\infty}$.
Thus, $f(\mathbf{x})\leq f(\left\Vert \mathbf{x}\right\Vert _{\infty}\vec{1}_{M})$
is a very mild form of monotonicity: {}``max-monotonicity''.
\end{remrk}

\begin{remrk}
\label{rem:powcon-stab-adj-qsemin-norms} All semi-norms and norms
satisfy conditions (\ref{eq:powcon-stab-adj-qsemin-pos}) , (\ref{eq:powcon-stab-adj-qsemin-subhom})
with equality, and (\ref{eq:powcon-stab-adj-qsemin-triang}) (see
Definitions \ref{def:math-norm-semi} and \ref{def:math-norm}). All
vector (semi-)norms that depend on the \emph{absolute value} of the
components of the vector --- such as the sub-family of H\"{o}lder
norms (Definition \ref{def:math-norm-p}) --- also satisfy condition
(\ref{eq:powcon-stab-adj-qsemin-mon}) (see Theorem \ref{thm:math-norm-mon-abs}). 
\end{remrk}

\subsection{Some immediate results}

\begin{lemma}
\label{lem:powcon-stab-adj-qsemi-triang-subhom} Suppose that $f:\Re^{M}\rightarrow\Re$
is such that $\lambda\in(0,1)\implies f(\lambda\mathbf{x})\leq\lambda f(\mathbf{x})$,
and $f(\mathbf{x}+\mathbf{y})\leq f(\mathbf{x})+f(\mathbf{y})$ then
$f$ satisfies $f(r\mathbf{x})\leq rf(\mathbf{x})\quad\forall\mathbf{x}\in\Re^{M}\mbox{ and }r\in\Re_{+}$
\end{lemma}
\IEEEproof 

Consider $m<r<m+1$ for $m\in\mathbb{N}$ (thus, $m$ is the {}``floor''
of $r$, $\lfloor r\rfloor$). Then $f(r\mathbf{x})\equiv f(m\mathbf{x}+(r-m)\mathbf{x})\leq f(m\mathbf{x})+f((r-m)\mathbf{x})$.
By Remark \ref{rem:powcon-stab-adj-qsemi-triang-subhom-nat} $f(m\mathbf{x})\leq mf(\mathbf{x})$.
By definition, $r-m\in(0,1)$, $\therefore f((r-m)\mathbf{x})\leq(r-m)f(\mathbf{x})$
by hypothesis. Hence, $f(r\mathbf{x})\leq mf(\mathbf{x})+(r-m)f(\mathbf{x})\equiv rf(\mathbf{x})$

\IEEEQED

\begin{remrk}
\label{rem:powcon-stab-adj-qsemi-triang-subhom} Lemma \ref{lem:powcon-stab-adj-qsemi-triang-subhom}
is valid for any $\mathbf{x}$, but we only need to apply it at the
point $\mathbf{x}=\vec{1}$ (i.e., $f(r\mathbf{x})\leq rf(\mathbf{x})\quad\forall r\in\Re_{+}\mbox{ at }\mathbf{x}=\vec{1}_{M}$). 
\end{remrk}

\begin{lemma}
\label{lem:powcon-stab-adj-qsemi-dotprod} Let $\mathbf{a}\in\Re^{M}$
with $a_{i}\neq0$. Then the function $f(\mathbf{x}):=\sum_{m=1}^{M}\vert a_{m}x_{m}\vert$
for $\mathbf{x}\in\Re^{M}$ satisfies Definition \ref{def:powcon-stab-adj-qseminormal}.
\end{lemma}
\begin{proof}
The relevant properties can be checked directly. Alternatively, one
may also write $f$ as $f(\mathbf{x})=\Vert D\mathbf{x}\Vert_{1}$
where $D$ is the diagonal matrix $D:=\mbox{diag}(a_{1},\cdots,a_{M})$,
and $\Vert\cdot\Vert_{1}$ denotes the H\"{o}lder 1-norm (Definition
\ref{def:math-norm-p}). Since $D$ is evidently non-singular, Theorem
\ref{thm:math-norm-matrix-X-vect} applies, and $f$ is a norm.
\end{proof}

\begin{lemma}
\label{lem:powcon-stab-adj-qsemi-macrod-apprx} For $\mathbf{x}\in\Re^{M}$
and $k=1,\cdots,K$, consider the vectors $\mathbf{a}_{k}=(a_{1,k},\cdots,a_{M,k})$
with $a_{m,k}\neq0$, and let $y_{k}(\mathbf{x})=\sum_{m=1}^{M}\vert a_{m,k}x_{m}\vert$,
$\mathbf{y}(\mathbf{x}):=(y_{1}(\mathbf{x}),\cdots,y_{K}(\mathbf{x}))$,
and $f(\mathbf{x}):=\Vert\mathbf{y}(\mathbf{x})\Vert_{\mu}$, where
$\Vert\cdot\Vert_{\mu}$ denote a monotonic norm on $\Re^{K}$ (see
Definition \ref{def:math-norm-mon}). Then $f$ satisfies Definition
\ref{def:powcon-stab-adj-qseminormal} .
\end{lemma}
\begin{proof}
By Lemma \ref{lem:powcon-stab-adj-qsemi-dotprod}, each $y_{k}$ can
be written as $y_{k}(\mathbf{x})=\Vert\mathbf{x}\Vert_{\nu_{k}}$where
$\Vert\cdot\Vert_{\nu_{k}}$ denotes a monotonic norm. Thus, $f$
can be written as $f(\mathbf{x})=\left\Vert \left[\begin{array}{ccc}
\Vert\mathbf{x}\Vert_{\nu_{1}} & \cdots & \Vert\mathbf{x}\Vert_{\nu_{K}}\end{array}\right]^{\prime}\right\Vert _{\mu}$.

By Theorem \ref{thm:math-norm-compos} ({}``norm of norms''), $f(\mathbf{x})$
is a norm.
\end{proof}

\subsection{Some examples}

\label{sec:powcon-stab-adj-props-exa}

\subsubsection{The simplest case}

\label{sub:powcon-stab-adj-props-exa-simp}

\begin{example}
\label{exa:powcon-stab-adj-qsemin-1} Consider a single-cell system,
and let $h_{j}p_{j}$ denote the \emph{received} power from terminal
$j$. Suppose that each terminal adjusts its power so that $h_{i}p_{i}/(Y_{i}(\mathbf{p}_{-i})+\sigma)=\alpha_{i}$,
where $Y_{i}=\sum_{\substack{n=1\\
n\neq i}
}^{N}h_{n}p_{n}$ is the interference affecting terminal $i$, and $\sigma$ represents
the average noise power. $p_{i}$ can be written as $f_{i}(\mathbf{p}_{-i})+c_{i}$,
with $f_{i}(\mathbf{p}_{-i}):=\sum_{\substack{n=1\\
n\neq i}
}^{N}(\alpha_{i}h_{n}/h_{i})\left|p_{n}\right|$ and $c_{i}=\sigma\alpha_{i}/h_{i}$. By Lemma \ref{lem:powcon-stab-adj-qsemi-dotprod}
, $f_{i}$ is a norm --- the absolute value operator has no real effect
here --- (Definition \ref{def:math-norm}), and hence has the desired
properties (see Remark \ref{rem:powcon-stab-adj-qsemin-norms})
\end{example}

\subsubsection{The macro-diversity scenario}

\label{sub:powcon-stab-adj-props-exa-macrod}

\paragraph{System model}

Under macro-diversity, the cellular structure is removed and each
transmitter is jointly decoded by all receivers\cite{powcon_phd_hanly93,powcon_macrod_cap_hanly96}.
A relevant QoS index for terminal $i$ is the product of its spreading
gain by its {}``carrier to interference ratio'' (CIR), $\alpha_{i}$,
defined as \cite{powcon_macrod_cap_hanly96} :\begin{equation}
\alpha_{i}=\frac{P_{i}h_{i,1}}{Y_{i,1}+\sigma_{1}^{2}}+\cdots+\frac{P_{i}h_{i,K}}{Y_{i,K}+\sigma_{K}^{2}}\label{eq:macrod-org-ai}\end{equation}

where $K$ is the number of receivers in the network, $h_{i,k}$ is
the channel gain in the signal from terminal $i$ arriving at receiver
$k$, and $Y_{i,k}$ denotes the interfering power experienced by
transmitter $i$ at receiver $k$; i.e., 

\begin{equation}
Y_{i,k}:=\sum_{\substack{n=1\\
n\neq i}
}^{N}P_{n}h_{n,k}\label{eq:macrod-org-Yik}\end{equation}

Below, we recognise and utilise the vectors: \begin{equation}
\mathbf{Y}_{i}:=(Y_{i,1},\cdots,Y_{i,K})\label{eq:macrod-Yi}\end{equation}
\begin{equation}
\boldsymbol{\sigma}:=(\sigma_{1}^{2},\cdots,\sigma_{K}^{2})\label{eq:macrod-noise-vect}\end{equation}

\paragraph{Normalised adjustment}

From (\ref{eq:macrod-org-ai}) one obtains the adjustment process
\begin{equation}
P_{i}=\alpha_{i}\left(\frac{h_{i,1}}{Y_{i,1}+\sigma_{1}^{2}}+\cdots+\frac{h_{i,K}}{Y_{i,K}+\sigma_{K}^{2}}\right)^{-1}\label{eq:macrod-org-adj}\end{equation}
It is unclear that the function on the right side of (\ref{eq:macrod-org-adj})
can be written as $f_{i}(\mathbf{p}_{-i})+c_{i}$ with $c_{i}\in\Re_{+}$and
$f_{i}$ satisfying Definition \ref{def:powcon-stab-adj-qseminormal}.
However, an adjustment rule that has the desired form, and  over estimates
the $P_{i}$ given by (\ref{eq:macrod-org-adj}) can be readily obtained. 

Reference \cite{powcon_macrod_cap_hanly96} simplifies the macro-diversity
analysis by including a terminal's own signal as part of the interference
(thus, the sum in equation (\ref{eq:macrod-org-Yik}) is taken over
\emph{all} $n$). As an alternative, in equation (\ref{eq:macrod-org-adj}),
one can replace each $Y_{i,k}(\mathbf{P})$ with \begin{equation}
\hat{Y}_{i}:=\max_{k}\{Y_{i,k}\}\equiv\Vert\mathbf{Y}_{i}\Vert_{\infty}\label{eq:macrod-apprx-interf-max}\end{equation}
and each $\sigma_{k}^{2}$ with \begin{equation}
\hat{\sigma}:=\max_{k}\{\sigma_{k}^{2}\}\equiv\left\Vert \boldsymbol{\sigma}\right\Vert _{\infty}\label{eq:macrod-apprx-noise-max}\end{equation}
Then, with \begin{equation}
h_{i}:=h_{i,1}+\cdots+h_{i,K}\label{eq:macrod-orig-hi}\end{equation}
equation (\ref{eq:macrod-org-adj}) becomes

\begin{equation}
P_{i}=\frac{\alpha_{i}}{h_{i}}\left(\hat{Y}_{i}+\hat{\sigma}\right)\label{eq:macrod-apprx-Pi}\end{equation}
Thus, the adjustment process can now be written as $P_{i}=f_{i}(\mathbf{P}_{-i})+c_{i}$
where,

\begin{equation}
f_{i}(\mathbf{P}_{-i}):=\frac{\alpha_{i}}{h_{i}}\left\Vert \mathbf{Y}_{i}(\mathbf{P}_{-i})\right\Vert _{\infty}\label{eq:macrod-apprx-adj-norm}\end{equation}
and\begin{equation}
c_{i}:=\frac{\alpha_{i}}{h_{i}}\hat{\sigma}\label{eq:macrod-apprx-ci}\end{equation}

\paragraph{Properties of the new macro-diversity adjustment}

\begin{prop}
\label{pro:powcon-stab-adj-exa-macrod} The function $f_{i}$ given
by equation (\ref{eq:macrod-apprx-adj-norm}) satisfies Definition
\ref{def:powcon-stab-adj-qseminormal}.
\end{prop}
\begin{proof}
In order to apply Lemma \ref{lem:powcon-stab-adj-qsemi-macrod-apprx},
let $\mathbf{x}:=\mathbf{P}_{-i}$ in such a way that $x_{n}=P_{n}$
for $n<i$ and $x_{n}=P_{n+1}$ for $n\geq i$). Likewise, let $a_{n,k}:=\alpha_{i}h_{n,k}/h_{i}$
for $n<i$ and $a_{n,k}:=\alpha_{i}h_{(n+1),k}/h_{i}$ for $n>i$
. (For example, for $N=3$ and $K=2$, if $i=2$, $x_{1}=P_{1}$,
$x_{2}=P_{3}$, $a_{1,k}=\alpha_{2}h_{1,k}/h_{2}$ and $a_{2,k}=\alpha_{2}h_{3,k}/h_{2}$).

The $k$th component of $(\alpha_{i}/h_{i})\mathbf{Y}_{i}(\mathbf{P})$
can then be written as $\sum_{m=1}^{N-1}\left|x_{m}\right|a_{k,m}\equiv\Vert\mathbf{x}\Vert_{\nu_{k}}$(see
Lemma \ref{lem:powcon-stab-adj-qsemi-dotprod}).

Thus, equation (\ref{eq:macrod-apprx-adj-norm}) can be written as
\begin{equation}
\left\Vert \left[\begin{array}{ccc}
\Vert\mathbf{x}\Vert_{\nu_{1}} & \cdots & \Vert\mathbf{x}\Vert_{\nu_{K}}\end{array}\right]^{\prime}\right\Vert _{\infty}\label{eq:macrod-apprx-adj-norm-gralform}\end{equation}
Lemma \ref{lem:powcon-stab-adj-qsemi-macrod-apprx}, with $\Vert\cdot\Vert_{\infty}$
playing the role of $\Vert\cdot\Vert_{\mu}$, implies that $f_{i}$
is a norm, and has, therefore, the desired properties (see Remark
\ref{rem:powcon-stab-adj-qsemin-norms}).
\end{proof}

\section{A fixed-point problem}

We seek to characterise the conditions leading to the convergence
of the process in which each terminal in a wireless communication
system, such as the reverse link of a CDMA cell, adjusts its transmission
power through a function of the form $f_{i}(\mathbf{p}_{-i})+c_{i}$
with $c_{i}\in\Re_{+}$ and $f_{i}$ satisfying Definition \ref{def:powcon-stab-adj-qseminormal}.

\subsection{Approach}

As discussed in subsection \ref{sub:powcon-stab-adj-disc-method},
we utilise fixed-point theory, in particular, Theorem \ref{thm:math-fixpt-banach1922},
the Banach Contraction Mapping principle.

\begin{remrk}
\label{rem:powcon-stab-adj-contr-norm} One can choose any metric
to apply Theorem \ref{thm:math-fixpt-banach1922}. Below we utilise
$d(\mathbf{x},\mathbf{y}):=\left\Vert \mathbf{x}-\mathbf{y}\right\Vert _{\infty}$
(see Definition \ref{def:math-norm-sup}), although the sub-index
of $\left\Vert \cdot\right\Vert _{\infty}$ is omitted for notational
convenience. 
\end{remrk}

\subsection{The Banach approach applied to our framework}

To apply fixed-point analysis, we need functions defined on $\Re^{N}$
.

\begin{lemma}
\label{lem:powcon-stab-adj-Rn-Rn1} For $x\in\Re^{N}$ let $g_{i}(\mathbf{x}):=0\cdot x_{i}+f_{i}(\mathbf{x}_{-i})\equiv f_{i}(\mathbf{x}_{-i})$.
If each $f_{i}$ satisfies Definition \ref{def:powcon-stab-adj-qseminormal}
as a function on $\Re^{N-1}$, then each $g_{i}$ satisfies Definition
\ref{def:powcon-stab-adj-qseminormal} as a function on $\Re^{N}$.
\end{lemma}
\begin{proof}
That $g_{i}$ has properties (\ref{eq:powcon-stab-adj-qsemin-pos})
and (\ref{eq:powcon-stab-adj-qsemin-mon}) follows trivially from
its definition and the hypothesis. 

To verify property (\ref{eq:powcon-stab-adj-qsemin-triang}), the
triangle inequality, notice that

$g_{i}(\mathbf{x}+\mathbf{y}):=f_{i}(\mathbf{x}_{-i}+\mathbf{y}_{-i})\leq f_{i}(\mathbf{x}_{-i})+f_{i}(\mathbf{y}_{-i})\equiv g_{i}(\mathbf{x})+g_{i}(\mathbf{y})$

To verify property (\ref{eq:powcon-stab-adj-qsemin-subhom}), sub-homogeneity,
observe that

$g_{i}(\lambda\mathbf{x}):=f_{i}(\lambda\mathbf{x}_{-i})\leq\lambda f_{i}(\mathbf{x}_{-i})+0\cdot x_{i}\equiv\lambda g_{i}(\mathbf{x})$
\end{proof}
\begin{thm}
\label{thm:powcon-stab-adj-contr-cond} Let $\vec{1}_{M}$ denote
the element of $\Re^{M}$ with each component equal to 1. For $\mathbf{x}\in\Re^{N}$
and $i\in\{1,\cdots,N\}$, let the transformation $\mathbf{T}$ be
defined by $T_{i}(\mathbf{x}):=f_{i}(\mathbf{x}_{-i})$ where each
$f_{i}$ satisfies Definition \ref{def:powcon-stab-adj-qseminormal}.
If $\forall i$, $f_{i}(\vec{1}_{N-1})<1$ then $\mathbf{T}$ is a
contraction (Definition \ref{def:math-contract-map}).
\end{thm}
\begin{proof}
For $x\in\Re^{N}$ let $g_{i}(\mathbf{x}):=0\cdot x_{i}+f_{i}(\mathbf{x}_{-i})\equiv f_{i}(\mathbf{x}_{-i})$.
By Lemma \ref{lem:powcon-stab-adj-Rn-Rn1}, each $g_{i}$ satisfies
Definition \ref{def:powcon-stab-adj-qseminormal} as a function on
$\Re^{N}$.

Let $\left\Vert \mathbf{T}(\mathbf{x})-\mathbf{T}(\mathbf{y})\right\Vert =$\begin{equation}
\left\Vert \left[\begin{array}{c}
g_{1}(\mathbf{x})-g_{1}(\mathbf{y})\\
\vdots\\
g_{N}(\mathbf{x})-g_{N}(\mathbf{y})\end{array}\right]\right\Vert =\max\left[\begin{array}{c}
\left|g_{1}(\mathbf{x})-g_{1}(\mathbf{y})\right|\\
\vdots\\
\left|g_{N}(\mathbf{x})-gf(\mathbf{y})\right|\end{array}\right]\label{eq:powcon-stab-adj-contr-explic}\end{equation}
By the reverse triangle inequality (see Lemma \ref{lem:math-triang-revrs}),
$\left|g_{i}(\mathbf{x})-g_{i}(\mathbf{y})\right|\leq g_{i}(\mathbf{x}-\mathbf{y})$.
Thus, \begin{equation}
\max\left[\begin{array}{c}
\left|g_{1}(\mathbf{x})-g_{1}(\mathbf{y})\right|\\
\vdots\\
\left|g_{N}(\mathbf{x})-g_{N}(\mathbf{y})\right|\end{array}\right]\leq\max\left[\begin{array}{c}
g_{1}(\mathbf{x}-\mathbf{y})\\
\vdots\\
g_{N}(\mathbf{x}-\mathbf{y})\end{array}\right]\label{ineq:powcon-stab-adj-contr-1}\end{equation}
Let $M_{x,y}:=\max(|x_{1}-y_{1}|,\cdots,|x_{N}-y_{N}|)\equiv\left\Vert \mathbf{x}-\mathbf{y}\right\Vert $

By monotonicity (condition (\ref{eq:powcon-stab-adj-qsemin-mon})),
\begin{equation}
g_{i}(\mathbf{x}-\mathbf{y})\leq g_{i}(M_{xy},\cdots,M_{xy})\equiv g_{i}(M_{xy}\vec{1}_{N})\label{ineq:powcon-stab-adj-contr-2}\end{equation}

By sub-homogeneity (condition (\ref{eq:powcon-stab-adj-qsemin-subhom}))\begin{equation}
g_{i}(M_{xy}\vec{1})\leq M_{xy}g_{i}(\vec{1})\equiv\left\Vert \mathbf{x}-\mathbf{y}\right\Vert g_{i}(\vec{1}_{N})\equiv\left\Vert \mathbf{x}-\mathbf{y}\right\Vert f_{i}(\vec{1}_{N-1})\label{ineq:powcon-stab-adj-contr-pre-fn}\end{equation}

Thus, \begin{equation}
\left\Vert \mathbf{T}(\mathbf{x})-\mathbf{T}(\mathbf{y})\right\Vert \leq\lambda\left\Vert \mathbf{x}-\mathbf{y}\right\Vert \label{ineq:powcon-stab-adj-contr-fn}\end{equation}
where $\lambda:=\max\{f_{1}(\vec{1}_{N-1}),\cdots,f_{N}(\vec{1}_{N-1})\}<1$.
\end{proof}
Therefore, with $f_{i}(\vec{1}_{N-1})<1$ for all $i$, the power
adjustment transformation is a contraction, and, by Theorem \ref{thm:math-fixpt-banach1922},
has a unique fixed point, which can be found by successive approximation.
Hence, a feasible power allocation exists that produces all the desired
QoS levels. When such allocation fails to exist, a reasonable course
of action is to proportionally reduce the QoS parameters \cite{powcon_cdma_capreg_mathar08}.

\section{Capacity implications}

Below, we will show how Theorem \ref{thm:powcon-stab-adj-contr-cond}
can be applied in the example scenarios of section \ref{sec:powcon-stab-adj-props-exa}.

\subsection{The simplest case}

In the scenario of section \ref{sub:powcon-stab-adj-props-exa-simp},
the adjustment rule is $f_{i}(\mathbf{p}_{-i})+c_{i}$, with $f_{i}(\mathbf{p}_{-i}):=\sum_{\substack{n=1\\
n\neq i}
}^{N}(\alpha_{i}h_{n}/h_{i})p_{n}$ and $c_{i}=\sigma\alpha_{i}/h_{i}$. 

The channel gains $h_{i}$ can be eliminated by working with the received
power levels, $P_{i}:=h_{i}p_{i}$. Now, each terminal adjusts its
power so that $P_{i}=\alpha_{i}(Y_{i}(\mathbf{P}_{-i})+\sigma)$ with
$Y_{i}=\sum_{\substack{n=1\\
n\neq i}
}^{N}P_{n}$. The adjustment rule can be re-written as $f_{i}(\mathbf{P}_{-i})+c_{i}$,
with $f_{i}(\mathbf{P}_{-i}):=\alpha_{i}\sum_{\substack{n=1\\
n\neq i}
}^{N}P_{n}$ and $c_{i}=\sigma\alpha_{i}$. 

The feasibility condition of Theorem \ref{thm:powcon-stab-adj-contr-cond}
requires that $\alpha_{i}\sum_{\substack{n=1\\
n\neq i}
}^{N}P_{n}<1$ with $P_{n}=1\quad\forall n$ . This leads to the eminently reasonable
condition: \begin{equation}
\alpha_{i}<1/(n-1)\label{eq:powcon-stab-adj-exa-simp-feas-first}\end{equation}

An alternate condition can be obtained through a simple coordinate
transformation. Let $q_{i}:=P_{i}/\alpha_{i}$, where $P_{i}$ denotes
\emph{received} power. Under the latest coordinates, the equivalent
adjustment is $q_{i}=g_{i}(\mathbf{q}_{-i})+\sigma$ with $g_{i}(\mathbf{q}_{-i}):=\sum_{\substack{n=1\\
n\neq i}
}^{N}q_{n}\alpha_{n}$. Now, the feasibility condition leads to \begin{equation}
\sum_{\substack{n=1\\
n\neq i}
}^{N}\alpha_{n}<1\label{eq:powcon-stab-adj-exa-simp-feas-best}\end{equation}
Condition (\ref{eq:powcon-stab-adj-exa-simp-feas-best}) is more flexible
than, and hence preferable to (\ref{eq:powcon-stab-adj-exa-simp-feas-first}),
because if the $\alpha_{i}$'s satisfy (\ref{eq:powcon-stab-adj-exa-simp-feas-first})
they automatically satisfy (\ref{eq:powcon-stab-adj-exa-simp-feas-best}),
but \emph{not} vice-versa.

\subsection{The macro-diversity scenario}

\label{sub:powcon-stab-adj-props-exa-macrod-cap}

\subsubsection{Original coordinates}

The feasibility condition of Theorem \ref{thm:powcon-stab-adj-contr-cond}
when applied to the adjustment rule of section \ref{sub:powcon-stab-adj-props-exa-macrod}
leads to (recall that $h_{i}=\sum_{k}h_{i,k}$):\begin{equation}
\alpha_{i}\sum_{\substack{n=1\\
n\neq i}
}^{N}\frac{h_{n,k}}{h_{i}}<1\quad\forall i,k\label{eq:powcon-stab-adj-exa-macrod-feas1}\end{equation}

\subsubsection{New coordinates}

\label{sub:powcon-stab-adj-props-exa-macrod-cap-best}

As with condition (\ref{eq:powcon-stab-adj-exa-simp-feas-first}),
condition (\ref{eq:powcon-stab-adj-exa-macrod-gik}) can be improved
upon through a change of coordinates. Equation (\ref{eq:macrod-apprx-Pi})
suggests the change of variable:\begin{equation}
q_{i}:=\frac{h_{i}P_{i}}{\alpha_{i}}\label{eq:powcon-stab-adj-exa-macrod-qi}\end{equation}
For convenience, let also \begin{equation}
g_{i,k}:=\frac{h_{i,k}}{h_{i}}\label{eq:powcon-stab-adj-exa-macrod-gik}\end{equation}
Now, $P_{n}h_{n,k}\equiv q_{n}\alpha_{n}h_{n,k}/h_{n}\equiv q_{n}\alpha_{n}g_{n,k}$.
Corresponding to equation (\ref{eq:macrod-org-Yik}), we now have\begin{equation}
Y_{i,k}:=\sum_{\substack{n=1\\
n\neq i}
}^{N}q_{n}\alpha_{n}g_{n,k}\label{eq:macrod-Yik(q)}\end{equation}

The adjustment process given by equation (\ref{eq:macrod-apprx-Pi})
can be expressed under the new coordinates, as $q_{i}=g_{i}(\mathbf{q}_{-i})+\hat{\sigma}$
with \begin{equation}
g_{i}(\mathbf{q}_{-i}):=\max_{k}\sum_{\stackrel{n=1}{n\neq i}}^{N}q_{n}\alpha_{n}g_{n,k}\equiv\left\Vert \mathbf{Y}_{i}(\mathbf{q}_{-i})\right\Vert _{\infty}\label{eq:macrod-apprx-adj-norm-new}\end{equation}
Now, the feasibility condition leads to\begin{equation}
\max_{i,k}\sum_{\substack{n=1\\
n\neq i}
}^{N}\alpha_{n}g_{nk}<1\label{eq:powcon-stab-adj-exa-macrod-feas-best}\end{equation}


\section{Non-sub-additive adjustment functions}

\label{sec:powcon-stab-adj-feas-extens}

Below we treat two cases: first the original adjustment rule is (sub)homogeneous
for any positive constant, a condition satisfied with equality by
all functions considered by \cite{powcon_stab_schubert05}. Then,
we consider specific models cited as examples by \cite{powcon_stab_yates95}.
The discussion in subsection \ref{sub:powcon-stab-adj-disc-other}
is important to this section.

\subsection{(sub)Homogeneous adjustment functions}

\label{sub:powcon-stab-adj-feas-xtens-boche}

Let us suppose that the original adjustment function fails to satisfy
the triangle inequality, but that, besides non-negative, it is monotonic,
and (sub)homogeneous for any positive constant. 

\begin{lemma}
\label{lem:powcon-stab-adj-contr-cond-shom-bound} Let $f:\Re^{M}\rightarrow\Re$
satisfy (i) non-negativity (\ref{eq:powcon-stab-adj-qsemin-pos}),
(ii) monotonicity (\ref{eq:powcon-stab-adj-qsemin-mon}), and (iii)
be such that $f(r\mathbf{x})\leq rf(\mathbf{x})\quad\forall\mathbf{x}\in\Re^{M}\mbox{ and }r\in\Re_{+}$.
Then there is a function $\phi:\Re^{M}\rightarrow\Re$ such that $f(\mathbf{x})\leq\phi(\mathbf{x})\quad\forall\mathbf{x}\in\Re^{M}$
and $\phi$ satisfies has. 
\end{lemma}
\begin{proof}
By monotonicity, $f(\mathbf{x})\leq f(\left\Vert \mathbf{x}\right\Vert _{\infty}\vec{1}_{M})$. 

By the sub-homogeneity hypothesis, \begin{equation}
f(\left\Vert \mathbf{x}\right\Vert _{\infty}\vec{1}_{M})\leq\left\Vert \mathbf{x}\right\Vert _{\infty}f(\vec{1}_{N})\label{eq:thm-contr-cond-subhom}\end{equation}
Thus, $f(\mathbf{x})\leq\left\Vert \mathbf{x}\right\Vert _{\infty}f(\vec{1}_{M})$.
$\phi$ defined by $\phi(\mathbf{x}):=\left\Vert \mathbf{x}\right\Vert _{\infty}f(\vec{1}_{M})$
has the desired properties.
\end{proof}

\begin{remrk}
\label{rem:powcon-stab-adj-feas-xtens-analys} $\phi(\mathbf{x})$
is just a scaled version of the infinity-norm $\left\Vert \cdot\right\Vert _{\infty}$
and hence satisfies Definition \ref{def:powcon-stab-adj-qseminormal}.
Thus, if each terminal adjusts its power with a function $f_{i}$
that satisfies non-negativity, monotonicity and (sub)homogeneity,
one can analyse the related system in which each terminal adjusts
its power with a corresponding $\phi_{i}(\mathbf{x}):=\left\Vert \mathbf{x}\right\Vert _{\infty}f_{i}(\vec{1})$. 
\end{remrk}

\begin{remrk}
\label{rem:powcon-stab-adj-feas-xtens-cap} By Theorem \ref{thm:powcon-stab-adj-contr-cond},
if $\phi_{i}(\vec{1})=\left\Vert \vec{1}\right\Vert f_{i}(\vec{1})\equiv f_{i}(\vec{1})\leq\lambda_{i}<1$,
the $\phi_{i}$-adjustment is asymptotically stable. And since each
$f_{i}$ satisfies $f_{i}(\mathbf{x})\leq\phi_{i}(\mathbf{x})$, one
can conclude that the {}``true'' adjustment process would behave
similarly, if the feasibility condition $f_{i}(\vec{1})\leq\lambda_{i}<1$
is satisfied.
\end{remrk}

\begin{remrk}
\label{rem:powcon-stab-adj-feas-xtens-bound-loose} There may exist
a different function, $\psi_{i}$, that satisfies Definition \ref{def:powcon-stab-adj-qseminormal},
and is such that $f_{i}(\mathbf{x})\leq\psi_{i}(\mathbf{x})\leq\phi_{i}(\mathbf{x})$
for all $\mathbf{x}\in\Re^{N-1}$. Indeed, the function we used to
{}``bound'' the original macro-diversity adjustment rule has the
more exotic {}``norm of norms'' form of eq.\ (\ref{eq:macrod-apprx-adj-norm-gralform}).
Thus, by exploiting the special structure of the original adjustment
function, if known, one may obtain a {}``tighter bound''. Nevertheless,
through Lemma \ref{lem:powcon-stab-adj-contr-cond-shom-bound} one
can obtain --- for a very large family of functions --- at least one
simple capacity result, when no better such result is available. 
\end{remrk}

\begin{remrk}
\label{rem:powcon-stab-adj-feas-xtens-bound-snorm} Additionally,
for $\mathbf{x}\in\Re^{N}$ and $1\leq p<q<\infty$ the H\"{o}lder
norms satisfy $\left\Vert \mathbf{x}\right\Vert _{\infty}\leq\left\Vert \mathbf{x}\right\Vert _{q}\leq\left\Vert \mathbf{x}\right\Vert _{p}\leq\left\Vert \mathbf{x}\right\Vert _{1}$\cite[Prop. 9.1.5, p. 345]{math_alg_lin_matrix_mg_bernstein05}.
This means that if any of these norms is to be used in the process
of building a bounding function for the original adjustment rule,
it should certainly be $\left\Vert \cdot\right\Vert _{\infty}$.
\end{remrk}

\subsection{Yates' framework}

\label{sub:powcon-stab-adj-feas-xtens-yates}

Below, we examine the specific scenarios given by \cite{powcon_stab_yates95}
as examples (the notation follows closely \cite{powcon_stab_yates95}).

\subsubsection{Scenarios studied in depth}

The power adjustment rule for fixed assignment, eq.\ (\cite{powcon_stab_yates95}-4),
can be written as $p_{j}=f_{j}(\mathbf{p})+c_{j}$ with $f_{j}(\mathbf{p})=(\gamma_{j}/h_{a_{j}j})\sum_{i\neq j}h_{a_{j}i}p_{i}$
and $c_{j}=\gamma_{j}\sigma_{a_{j}}/h_{a_{j}j}$. $f_{j}$ is a norm
(see Lemma \ref{lem:powcon-stab-adj-qsemi-dotprod}) and hence satisfies
Definition \ref{def:powcon-stab-adj-qseminormal}. Thus, this case
perfectly fits our formulation, and in fact is closely related to
the simple example discussed in subsection \ref{sub:powcon-stab-adj-props-exa-simp}.

Likewise, the full macro-diversity model has already been fully addressed,
and in fact, a corresponding new capacity result been found and discussed
(see subsection \ref{sub:powcon-stab-adj-disc-macrod} for a summary).

\subsubsection{Other scenarios}

The remaining examples of \cite{powcon_stab_yates95} can be easily
handled by neglecting random noise. It is straightforward to verify
that, if one neglects noise, the corresponding power adjustment rules
are homogeneous of degree one, and hence fall under the analysis of
subsection \ref{sub:powcon-stab-adj-feas-xtens-boche}. Below we shall
discuss in greater detail the case of multiple-connection (MC) reception.
This is an interesting and challenging model which contains another
scenario, the minimum power assignment (MPA), as a special case.

\subsubsection{The MC scenario}

Under MC, user $j$ must maintain an acceptable SIR $\gamma_{j}$
at $d_{j}$ distinct base stations. The system {}``assigns'' $j$
to the $d_{j}$ {}``best'' receivers. Let $Y_{kj}(\mathbf{p}):=\sum_{i\neq j}h_{ki}p_{i}$
and suppose there are $K$ receivers. For $x\in\Re_{+}^{M}$ and $m\leq M$,
let $\max(x;m)$ and $\min(x;m)$ denote, respectively, the $m$th
largest and the $m$th smallest component of $x$. The requirements
of $j$ can be written as $\max\left(\left(p_{j}h_{1j}/(Y_{1j}+\sigma_{1}),\cdots,p_{j}h_{Kj}/(Y_{Kj}+\sigma_{K})\right);d_{j}\right)\geq\gamma_{j}$
or, equivalently, as \cite{powcon_stab_yates95}:\begin{equation}
p_{j}\geq\gamma_{j}\min\left(\left(\frac{Y_{1j}(\mathbf{p})+\sigma_{1}}{h_{1j}},\cdots,\frac{Y_{Kj}(\mathbf{p})+\sigma_{K}}{h_{Kj}}\right);d_{j}\right)\label{eq:powcon-stab-adj-feas-extens-yates-mc-constr}\end{equation}

Under the mild assumption that $\sigma_{k}\ll Y_{kj}$ $\forall k$
and hence can be dropped, the right side of (\ref{eq:powcon-stab-adj-feas-extens-yates-mc-constr})
is clearly homogeneous of degree one in $\mathbf{p}$. Hence, the
discussion of subsection \ref{sub:powcon-stab-adj-feas-xtens-boche}
applies to this case. Proceeding as in subsection \ref{sub:powcon-stab-adj-props-exa-macrod-cap-best},
we apply condition $f_{j}(\vec{1})<1$ to a slightly different form
of (\ref{eq:powcon-stab-adj-feas-extens-yates-mc-constr}) in which
the variables are $q_{j}=p_{j}/\gamma_{j}$, for which $Y_{kj}(\mathbf{q}):=\sum_{i\neq j}h_{ki}\gamma_{i}q_{i}$.
This leads to the condition: \begin{equation}
\min\left(\left(\sum_{i\neq j}\frac{h_{1i}}{h_{1j}}\gamma_{i},\cdots,\sum_{i\neq j}\frac{h_{Ki}}{h_{Kj}}\gamma_{i}\right);d_{j}\right)<1\,\,\,\forall j\label{eq:powcon-stab-adj-feas-xtens-yates-mc-fin}\end{equation}

This condition involves weighted sums of $N-1$ quality-of-service
parameters where the weights are relative channel gains. For instance,
with $d_{j}=3$, condition (\ref{eq:powcon-stab-adj-feas-xtens-yates-mc-fin})
requires that the 3rd smallest such sum be less than one. 

Condition (\ref{eq:powcon-stab-adj-feas-xtens-yates-mc-fin}) has
similarities with (\ref{eq:powcon-stab-adj-exa-macrod-feas-best}),
its macro-diversity counterpart. But the relative gains are not defined
in the same way ($h_{ki}/h_{kj}$ in (\ref{eq:powcon-stab-adj-feas-xtens-yates-mc-fin}),
versus $h_{ki}/\sum_{k}h_{ki}$ in (\ref{eq:powcon-stab-adj-exa-macrod-feas-best})). 

In fact, one can apply here the same simplification used for macro-diversity
in subsection \ref{sub:powcon-stab-adj-props-exa-macrod}. Let $\mathbf{Y}_{j}(\mathbf{p}):=(Y_{1j}(\mathbf{p}),\cdots,Y_{Kj}(\mathbf{p}))$,
$\boldsymbol{\sigma}:=(\sigma_{1}^{2},\cdots,\sigma_{K}^{2})$, and
$\mathbf{H_{j}}:=(h_{1j},\cdots,h_{Kj})$. Then, replace each $Y_{kj}(\mathbf{p})$
with $\hat{Y}_{j}:=\max_{k}\{Y_{kj}\}\equiv\left\Vert \mathbf{Y}_{j}\right\Vert _{\infty}$
and each $\sigma_{k}^{2}$ with $\hat{\sigma}:=\max_{k}\{\sigma_{k}^{2}\}\equiv\left\Vert \boldsymbol{\sigma}\right\Vert _{\infty}$.
The requirements of user $j$ can now be written as $p_{j}\max(\mathbf{H_{j}};d_{j})/\bigl(\hat{Y}_{j}+\hat{\sigma}\bigr)\geq\gamma_{j}$,
which, with $h_{j}:=\max(\mathbf{H_{j}};d_{j})$, leads to the adjustment
$p_{j}h_{j}/\gamma_{j}=\hat{Y}_{j}+\hat{\sigma}$, or equivalently
to:

\begin{equation}
q_{j}\equiv\left\Vert \mathbf{Y}_{j}(\mathbf{q})\right\Vert _{\infty}+\hat{\sigma}\label{eq:powcon-stab-adj-feas-xtens-yates-mc-adj-new}\end{equation}

where $q_{j}:=p_{j}h_{j}/\gamma_{j}$, $Y_{kj}(\mathbf{q})=\sum_{i\neq j}\gamma_{i}q_{i}g_{ki}$
and $g_{ki}:=h_{ki}/h_{i}$. 

This leads to the feasibility condition\begin{equation}
\max_{j,k}\sum_{i\neq j}\gamma_{i}g_{ki}<1\label{eq:powcon-stab-adj-feas-xtens-yates-mc-fin2}\end{equation}

Condition (\ref{eq:powcon-stab-adj-feas-xtens-yates-mc-fin2}) is
virtually identical to (\ref{eq:powcon-stab-adj-exa-macrod-feas-best}).
$g_{ki}:=h_{ki}/h_{i}$ in both cases. However, in (\ref{eq:powcon-stab-adj-exa-macrod-feas-best})
$h_{i}:=\sum_{k}h_{ki}$, whereas in (\ref{eq:powcon-stab-adj-feas-xtens-yates-mc-fin2})
$h_{i}:=\max((h_{1i},\cdots,h_{Ki});d_{i})$ (e.g., if $d_{i}=3$,
the corresponding $h_{i}$ is the third highest of $i$'s channel
gains). 

Notice that both conditions (\ref{eq:powcon-stab-adj-feas-xtens-yates-mc-fin2})
and (\ref{eq:powcon-stab-adj-feas-xtens-yates-mc-fin}) underestimate
the capacity of the MC system, but for different reasons. Further
work may determine which condition is more advantageous. 


\appendices

\section{Norms, metrics and related material}

\setcounter{equation}{0} 

\setcounter{definitn}{0} 

\setcounter{remrk}{0} 

\setcounter{prop}{0} 

\numberwithin{prop}{section}

\numberwithin{definitn}{section}

\numberwithin{remrk}{section}

\renewcommand{\theequation}{\thesection.\arabic{equation}}

\label{sec:math-analys-background}

\subsection{Concepts and definitions}

Let $V$ denote a vector space (for a formal definition see \cite[pp. 11-12]{opt_func_luenberger69}). 

\begin{definitn}
\label{def:math-norm-semi}A function $f\colon V\to\Re$ is called
a \emph{semi-norm} on $V$, if it satisfies: 
\end{definitn}
\begin{enumerate}
\item \label{enu:math-norm-posit} $f(x)\geq0$ for all $x\in V$
\item \label{enu:math-norm-homog} $f(\lambda x)=\vert\lambda\vert\cdot f(x)$
for all $x\in V$ and all $\lambda\in\Re$ (homogeneity)
\item \label{enu:math-norm-triang} $f(x+y)\leq f(x)+f(wy)$ for all $x,y\in V$
(the \emph{triangle inequality}) 
\end{enumerate}

\begin{definitn}
\label{def:math-norm} If a semi-norm additionally satisfies $f(x)=0$
if and only if $x=\theta$ (where $\theta$ denotes the zero element
of $V$), then $f$ is called a norm on $V$ and $f(x)$ is usually
denoted as $\Vert x\Vert$. 
\end{definitn}

\begin{remrk}
\label{rem:math-norm-convex} It is a simple matter to show that a
function that satisfies properties \ref{enu:math-norm-homog} and
\ref{enu:math-norm-triang} above is convex. Thus, (semi-)norm-minimisation
problems are often well-behaved. 
\end{remrk}

\begin{definitn}
\label{def:math-norm-p} The H\"{o}lder norm with parameter $p\geq1$
({}``$p$-norm'') is denoted as $||\cdot||_{p}$ and defined for
$x\in\Re^{N}$ as $\left\Vert \mathbf{x}\right\Vert _{p}=(|x_{1}|^{p}+\cdots+|x_{N}|^{p})^{\frac{1}{p}}$.
\end{definitn}
\begin{remrk}
\label{rem:math-norm-p-special} With $p=2$, the H\"{o}lder norm
becomes the familiar Euclidean norm. The $p=1$ case is also often
encountered (see Lemma \ref{lem:powcon-stab-adj-qsemi-dotprod}).
Furthermore, it can be shown that $\lim_{p\rightarrow\infty}\left\Vert \mathbf{x}\right\Vert _{p}=\max(\left|x_{1}\right|,\cdots,\left|x_{N}\right|)$,
which leads to the following definition:
\end{remrk}
\begin{definitn}
\textbf{\label{def:math-norm-sup} }For $x\in\Re^{N}$, the supremum
or infinity norm is denoted as $\left\Vert \cdot\right\Vert _{\infty}$
and defined as\begin{equation}
\left\Vert \mathbf{x}\right\Vert _{\infty}:=\max(\left|x_{1}\right|,\cdots,\left|x_{N}\right|)\label{eq:macrod-infinity-norm}\end{equation}

\end{definitn}

\begin{definitn}
\label{def:math-vect-abs} For $\mathbf{x}\in\Re^{N}$ denote as $|\mathbf{x}|$
the vector whose $i$th component is obtained as the absolute value
of the $i$th component of $\mathbf{x}$, $|x_{i}|$. 
\end{definitn}

\begin{definitn}
\label{def:math-norm-abs} A norm, $\Vert\cdot\Vert$, on $\Re^{N}$
is called an \emph{absolute vector norm} if it depends only on the
absolute values of the components of the vector; that is, for $\mathbf{v}\in\Re^{N}$,
and $\mathbf{w}:=|\mathbf{v}|$, $\Vert\mathbf{v}\Vert\equiv\Vert\mathbf{w}\Vert$. 
\end{definitn}

\begin{definitn}
\label{def:math-norm-mon} For $\mathbf{x}\mbox{ and }\mathbf{y}\in\Re^{N}$,
let $\mathbf{x}\leq\mathbf{y}$ mean that $x_{i}\leq y_{i}$ for each
$i$. A norm, $\Vert\cdot\Vert$, on $\Re^{N}$ is said to be \emph{monotonic}
if, for any $\mathbf{x}\mbox{ and }\mathbf{y}\in\Re^{N}$, $|\mathbf{x}|\leq|\mathbf{y}|$
implies that $\Vert\mathbf{x}\Vert\leq\Vert\mathbf{y}\Vert$. 
\end{definitn}

\begin{definitn}
\label{def:math-metric} A \emph{metric}, or \emph{distance} function
is a real valued function $d:X\times X\longrightarrow\Re$ where $X$
is some set, such that, for every $x,y,z\in X$, (i) $d(x,y)\geq0$,
with equality if and only if $x=y$ , (ii) $d(x,y)=d(y,x)$ and (iii)
$d(x,z)\leq d(x,y)+d(y,z)$ (the \emph{triangle inequality})
\end{definitn}
\begin{remrk}
\label{rem:math-norm-metric} Every norm $\left\Vert \cdot\right\Vert $
on a vector space $V$ engenders the metric $d(x,y)=\left\Vert x-y\right\Vert $
for $x\,,\, y\,\in V$. A norm generalises the intuitive notion of
size or length, while a metric generalises the intuitive notion of
distance.
\end{remrk}
\begin{definitn}
\label{def:math-space-metric} A \emph{metric space} $(X,d)$ is a
set $X$, together with a \emph{metric} $d$ defined on $X$. If every
\emph{Cauchy} sequence of points in $X$ has a limit that is also
in $X$ then $(X,d)$ is said to be \emph{complete}.
\end{definitn}

\subsection{Useful results from the literature}

\begin{lemma}
\label{lem:math-triang-revrs} (Reverse triangle inequality) If the
function $f\colon V\to\Re$ satisfies the triangle inequality, then
$\vert f(\mathbf{x})-f(\mathbf{y})\vert\leq f(\mathbf{x}-\mathbf{y})$. 
\end{lemma}
\begin{proof}
Without loss of generality, suppose that $f(\mathbf{x})\geq f(\mathbf{y})$
which implies that $f(\mathbf{x})-f(\mathbf{y})\equiv\vert f(\mathbf{x})-f(\mathbf{y})\vert$.

Observe that $\mathbf{x}\equiv(\mathbf{x}-\mathbf{y})+\mathbf{y}$
and apply the triangle inequality to this sum:

Thus, $f(\mathbf{x})\equiv f((\mathbf{x}-\mathbf{y})+\mathbf{y})\leq f(\mathbf{x}-\mathbf{y})+f(\mathbf{y})$
or \begin{equation}
f(\mathbf{x})-f(\mathbf{y})=\vert f(\mathbf{x})-f(\mathbf{y})\vert\leq f(\mathbf{x}-\mathbf{y})\label{eq:math-triang-revrs}\end{equation}

\end{proof}

\begin{remrk}
\label{rem:math-triang-revrs-cont} Through (\ref{eq:math-triang-revrs})
one can prove that all norms are continuous.
\end{remrk}
\begin{thm}
\label{thm:math-norm-mon-abs} A norm on $\Re^{N}$is monotonic if
and only if it is an absolute vector norm.
\end{thm}
\begin{proof}
See \cite{math_analys_norm_monot_bauer61} or \cite[p.344]{math_alg_lin_matrix_mg_bernstein05}.
\end{proof}
\begin{thm}
\label{thm:math-norm-compos} ({}``Norm of norms''). Let $\Vert\cdot\Vert_{\nu_{1}},\cdots,\Vert\cdot\Vert_{\nu_{M}}$
be $M$ given vector norms on a real (or complex) vector space $V$,
and let $\Vert\cdot\Vert_{\mu}$ be a \emph{monotonic} vector norm
on $\Re^{M}$. Then, $\Vert\mathbf{x}\Vert:=\left\Vert \left[\Vert\cdot\Vert_{\nu_{1}},\cdots,\Vert\cdot\Vert_{\nu_{M}}\right]^{T}\right\Vert _{\mu}$
is a \emph{norm}.
\end{thm}
\begin{proof}
See \cite[Theorem 5.3.1]{math_alg_lin_matrix_mg_horn85}.
\end{proof}

\begin{thm}
\label{thm:math-norm-matrix-X-vect} Let $\Vert\cdot\Vert$ be a monotonic
norm on $\Re^{M}$ and let $T$ be an $M\times M$ non-singular real
matrix. Then, $\Vert\mathbf{x}\Vert_{T}:=\Vert T\mathbf{x}\Vert$
for $\mathbf{x}\in\Re^{M}$ defines another\emph{ }monotonic norm
on $\Re^{M}$. 
\end{thm}
\begin{proof}
See \cite[Theorem 5.3.2]{math_alg_lin_matrix_mg_horn85}.
\end{proof}

\section{Banach fixed-point theory}

\label{sec:math-fixpt-results} 

\begin{definitn}
\label{def:math-contract-map} A map $T$ from a metric space $(X,d)$
into itself is a \emph{contraction} if there exists $\lambda\ \in[0,1)$
such that for all $x\mbox{\ ,\ }y\ \in V$, $d(T(x),T(y))\leq\lambda d(x,y)$. 
\end{definitn}

\begin{definitn}
\label{def:math-succes-approx} Picard iterates (Successive approximation):
Let $T^{m}(x_{1})$ for $x_{1}\in V$ be defined inductively by $T^{0}(x_{1})=x_{1}$
and $T^{m+1}(x_{1})=T\left(T^{m}(x_{1})\right)$, with $m\in\left\{ 1,2,\cdots\right\} $.
\end{definitn}
\begin{thm}
\label{thm:math-fixpt-banach1922} (Banach' Contraction Mapping Principle)
If $T$ is a contraction mapping on a complete metric space $(X,d)$
then there is a unique $x^{*}\in X$ such that $x^{*}=T(x^{*})$.
Moreover, $x^{*}$ can be obtained by successive approximation, starting
from an arbitrary initial $x_{0}\in X$ ; i.e., for any $x_{0}\in X$,
$\lim_{m\rightarrow\infty}T^{m}(x_{0})=x^{*}$.
\end{thm}
\begin{proof}
See \cite{math_analys_fixpt_phd_banach1920}\cite[Theorem 3.1.2, p. 74]{math_analys_fixpt_istratescu81}.
\end{proof}

\


\bibliographystyle{IEEEtran}
\bibliography{powcon_stab09}

\end{document}